\documentclass[review]{elsarticle}

\usepackage{amsfonts}
\usepackage{graphicx}
\usepackage{algorithmicx,algorithm,algpseudocode}
\usepackage{algorithm}
\usepackage{algpseudocode}
\usepackage{subcaption}
\usepackage{placeins}
\usepackage{multirow}
\usepackage{amsmath}
\usepackage{amsthm,amssymb}
\usepackage{mathtools}
\usepackage{lineno,hyperref}
\modulolinenumbers[5]
\usepackage{cleveref}

\DeclareMathOperator{\polylog}{polylog}
\DeclareMathOperator{\MPC}{MPC_n}

\newtheorem{lemma}{Lemma}
\newtheorem{theorem}{Theorem}
\begin{document}

\begin{frontmatter}

\title{An Efficient Construction of Yao-Graph in Data-Distributed Settings}

\author[sharif]{Sepideh Aghamolaei\corref{corresponding}}%fnref
\address[sharif]{Department of Computer Engineering, Sharif University of Technology, Tehran, Iran}
%\cortext[s]{s}
\ead{sepideh.aghamolaei14@sharif.edu}
%\fntext[emailsharif]{aghamolaei@ce.sharif.edu, ghodsi@sharif.edu}

\author[sharif,ipm]{Mohammad Ghodsi}
\ead{ghodsi@sharif.edu}
\address[ipm]{School of Computer Science, Institute for Research in Fundamental Sciences (IPM), Iran}
\cortext[corresponding]{Corresponding author}

\begin{abstract}
A sparse graph that preserves an approximation of the shortest paths between all pairs of points in a plane is called a geometric spanner.
%Geometric spanners are used to speed up the computation of shortest paths.
Using range trees of sublinear size, we design an algorithm in massively parallel computation (MPC) model for constructing a geometric spanner known as Yao-graph. This improves the total time and the total memory of existing algorithms for geometric spanners from subquadratic to near-linear.
\end{abstract}


\end{frontmatter}

\section{Introduction}
Consider a graph with vertices that are a set of points in the Euclidean plane and edges that are all the pairs of vertices weighted by their Euclidean distance.
Keeping a near-linear size subset of the edges of a geometric graph such that the distances between all pairs of vertices are approximately preserved results in a subgraph called a geometric spanner. In this sparsification method for all-pairs shortest paths, the ratio of the shortest path in the geometric spanner to the shortest path in the original graph is called the stretch factor.

Formally, a $\epsilon$-spanner of $G$ is a subgraph $H$ of $G$ such that the shortest path between any pair of vertices (nodes) in $H$ approximates the shortest path between any pair of vertices in $G$ by $1+\epsilon$, for a given $\epsilon>0$. A geometric spanner uses the complete Euclidean graph (the graph that connects all pairs of vertices using edges weighted by the distances between those vertices) as $G$.

A theoretical model for cloud processing frameworks such as MapReduce and Spark is the massively parallel computation model ($\mathbf{MPC}$)~\cite{mpcj}  that restricts the memory of each machine to $m=\Theta(n^{\eta})$ and the number of machines to $L=O(n/m)$, for a constant $\eta\in(0,1)$ and an input of size $n$.
The complexity measures are the number of rounds $R(n)$ and the amount of communication $C(n)$, and the condition $R(n)=O(\polylog(n))$ and $C(n)=O(nR(n))$ must also hold.
We denote this by $\MPC(m)$. With a slight abuse of notation, we use $\mathbf{MPC}$ as the class of the problems with algorithms in this model.
A closely related model is MapReduce class (MRC)~\cite{mrc} which relaxes the condition $L=O(n/m)$ to $L=O(n^{\eta'})$ for a constant $\eta'\in(0,1)$.

$\mathbf{MPC}$ and MRC are contained in $\mathbf{P}$, more specifically, subquadratic and near-linear time algorithms, respectively, and they both contain $\mathbf{NC}^1$ (the class of problems solvable using a PRAM machine with linear memory in logarithmic time)~\cite{goodrich2}.
For a set of values $x_1,\ldots,x_n$ and an associative binary operator $\oplus$, the parallel semi-group computation $x_1\oplus\cdots\oplus x_n$ and the parallel prefix computation $x_1\oplus \cdots \oplus x_i$, for all $i=1,\ldots,n$ can be computed in $O(\log_m n)$ rounds in MPC~\cite{goodrich2,frei2019efficient}. Parallel semi-group is useful for aggregations such as summation, minimum, and maximum computation and parallel prefix can be used for dynamic programs on trees with constant depth.
Any algorithm with $O(\log^{i+1} n)$ running time in PRAM has $O(\log^i )$ rounds in MRC, for $i\geq 1$~\cite{frei2019efficient}. So, the algorithm for WSPD-spanner in PRAM with $O(\log^2 n)$ time takes $O(\log n)$ rounds in MRC~\cite{callahan1995decomposition}.

Our main contributions are:
\begin{itemize}
\item designing an algorithm for constructing a version of the range tree in $\mathbf{MPC}$ that we call a sparsified range tree (SRT), and
\item spanners built inside squares of sublinear size can be augmented using Yao spanner edges in $\mathbf{MPC}$.
\end{itemize}

\Cref{table:spanners} shows a summary of results.
Note that since $\mathbf{MPC}$ is more restricted than MRC, all $\mathbf{MPC}$ algorithms work in MRC, too.

\begin{table}[h]
\centering
\begin{tabular}{|c|c|c|c|c|}
\hline
Spanner & Model & Size & Rounds & References\\
\hline
\hline
Yao-Graph & MRC & $O(n)$ & $O(1)$ & \cite{mine} \\
& $\mathbf{MPC}$ & $O(n)$ & $O(1)$ & \Cref{sec:yao}\\
\hline
\end{tabular}
\caption{Yao-Graph in MapReduce Models}
\label{table:spanners}
\end{table}

\section{Preliminaries}
\subsection{A Balanced Grid in $\mathbf{MPC}$}
To use a balanced grid with $m$ points in each cell in $\mathbf{MPC}$~\cite{mine}, build a grid $G$ by the split lines in each dimension, then, merge every $\sqrt{m}$ cell in each dimension. Map the points to the cells of the resulting grid to partition them.

\subsection{Range Tree in Sequential and PRAM Models}
Range tree~\cite{bentley1979decomposable} is a tree data structure for querying the set of points inside rectangular ranges. The construction of 1D range tree builds a binary search tree on one of the dimensions and stores the maximum of the left subtree in each node. For higher dimensions, first, a 1D tree is built for the first dimension, then, for each internal node, a range tree is built for the next dimension, and this continues until all dimensions have trees. The construction of range tree on $n$ points in $d$-dimensional space takes $O(n\log^d n)$ time and the query time for the number of points inside the range using this tree is $O(\log^d n)$. Range tree in 2D in PRAM can be built in $O(\log^2 n)$ time using $O(n \log n)$ processors~\cite{lopez2020parallel}.
\subsubsection{2-Sided Range Queries for $\ell_1$ Nearest Neighbor}
Given a point set $P$, $2$-sided range queries take a rectangular shape $B$ and a corner $q=(x_0,y_0)$ of $B$ as their input, nearest neighbor queries in $\ell_1$ find the closest point to $q$ in $P$ using the $\ell_1$ distance $(\lVert p-q\rVert)$, where the $\ell_1$ norm of a vector $\tau=(\tau_1,\ldots,\tau_d)$ is defined as $\lVert \tau \rVert_1=\sum_{i=1}^d \lvert \tau_i \rvert$.
The region defined by one of the four following cases is called a $2$-sided range~\cite{arge2001external}:
 $x\geq x_0, y\geq y_0$,
 $x\geq x_0, y\leq y_0$,
 $x\leq x_0, y\geq y_0$, and
 $x\leq x_0, y\leq y_0$.

The nearest neighbor of the point with the $i$-th and the $j$-th smallest coordinates is denoted by $NN[i][j]$ and the set of points in the cell that corresponds to the range $[x_i,x_{i+1}]\times[y_j,y_{j+1}]$ is denoted by $T[i][j]$. The recursive formula for the nearest neighbor in the set of cones that are parallel to the positive directions of the axes is
\[
NN[i][j]=\min (T[i][j], NN[i-1][j], NN[i][j-1]),
\]
\[
T[i][j]=\begin{cases}
\emptyset &(x_i,y_j)\ni P\\
(x_i,y_j) &(x_i,y_j)\in P\\
\end{cases}.
\]

\begin{lemma}[\cite{mine}]\label{lemma:semigroup}
There is an $O(1)$-round algorithm for $2$-sided and rectangular range queries in the massively parallel computations model that solves $\ell_1$ nearest neighbor queries on $n$ points (set $P$).
\end{lemma}
Segment tree for storing and answering range queries in $\mathbf{MPC}$ exists~\cite{jdsa}.
\subsubsection{Yao-graph in Sequential and $\mathbf{MRC}$ Models}
For a set of $n$ points, the Yao-graph spanner with $k$ cones partitions the angle around each point into $k$ equal cones such that the first side of the first cone is parallel to the $x$ axis; Then, each point is connected to its nearest neighbor among the points that are inside each cone.
The size (the number of edges) of this spanner is $O(kn)$, its stretch factor is $\frac{1}{\cos(2\pi/k)-\sin(2\pi/k)}$ for $k\geq 9$, and it can be constructed in $O(n\log n)$ time~\cite{bose2004approximating}. These bounds have been improved and Yao-graph for $k\geq 4$ is also a spanner~\cite{spanner}.
We review a result from~\cite{mine}.

\begin{lemma}[\cite{mine}]\label{lemma:mine}
The stretch factor of the resulting spanner is $1+\frac{\theta^2}{8}+o(\theta^2)$ for the Euclidean distance.
\end{lemma}

\begin{figure}[h]
\centering
\begin{subfigure}{0.45\textwidth}
\centering
\includegraphics[scale=0.8]{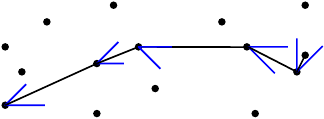}
\caption{Original method (sequential model)}
\end{subfigure}
\hfill
\begin{subfigure}{0.45\textwidth}
\centering
\includegraphics[scale=0.8]{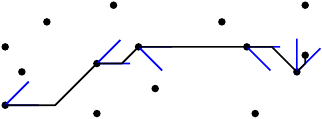}
\caption{Mapping from $\ell_1$ (MRC model)}
\end{subfigure}
\caption{A shortest path in Yao-graph with $8$ cones around each point $(k=8)$.}
\end{figure}
%%%%%%%%%%%%%%%%

\section{Yao-Graph in $\mathbf{MPC}$}\label{sec:spanner}

First, we introduce a new tree data structure for $2$-sided range queries and use it to approximate a Yao-graph spanner with near-linear size.
\subsection{Sparsified Range Tree (SRT)}
For windowing range queries, we build a multi-dimensional range tree and use its cells to define the queries in the first two dimensions, while preserving other properties in the rest of the dimensions.
Remember that a range tree is constructed in a bottom-up manner. The recursive computation used in this construction is a semi-group computation and is therefore computable via a $m$-ary recursion tree~\cite{goodrich2}. So, the $\mathbf{MPC}$ simulation of range tree takes $O(\log_m n)=O(1/\eta)$ rounds and $O(n\log_m n)=O(n/\eta)$ space.

\begin{algorithm}[h]
\caption{Sparsified Range Tree in $\MPC(m)$}
\label{alg:flat}
\begin{algorithmic}[1]
\Require{The set of points $P$ in $\mathbb{R}^d$}
\Ensure{A SRT of $P$}
\For{$i=1$ to $d$}
\State{$S_i=$ elements of rank $\{jm\mid j=1,\ldots,\lfloor \frac{n}{m}\rfloor\}$ from $P$.}
\EndFor
\State{$G\gets$ the $d$-dimensional grid with cells $S_1\times S_2\times \cdots \times S_{d}$ and assign the points of $P$ to the cells containing them.}
\State{Build a range tree $T$ with the cells of $G$ as its leaves and define the rest of $T$ recursively by merging every $\sqrt[d]{m}$ consecutive node in each dimension.}
\\ \Return $T$
\end{algorithmic}
\end{algorithm}
\Cref{alg:flat} computes a SRT by first creating a grid, then, it indexes them based on a range tree built on the grid, and finally, it merges the nodes of the tree similar to a B-tree to reduce the height of the tree. The optional last step can aggregate the results to get approximate results if an objective function is given.
Summing values on a tree is a semi-group computation, so, it takes $O(\log_m n)=O(\log n / \log m)=O(1/\eta)$ rounds. There are $d+2$ such computations in the algorithm, resulting in $O(d/\eta)$ rounds in total.

\begin{lemma}
A SRT of $n$ points has $O(1)$ depth and $O(n/m)$ nodes in $\MPC(m)$.
\end{lemma}
\begin{proof}
Each node has $O(m)$ children since there are $O(\sqrt[d]{m})$ points in each dimension. So, the tree has $O(n/m)$ nodes and $T$ is a $m$-ary tree.
The depth of the tree is $O(\log_m n)=O(\log n / \log (n^{\eta}))=O(1/{\eta})=O(1)$ since $m=\Theta(n^{\eta})$ in $\MPC(m)$.
\end{proof}

\subsection{2-Sided Range Nearest Neighbor Query using SRT}
We give an algorithm for nearest neighbors using $2$-sided range queries that performs a parallel prefix on a SRT. The proof of the associativity of the nearest neighbor inside a cone using $\ell_1$ distance exists~\cite{mine} which is needed for the parallel prefix computation.
We approximate the Euclidean nearest neighbor by the nearest neighbor using $\ell_1$ distance.
The dimensions of the tree are the Cartesian coordinates $x$ and $y$, so each node of the tree defines four $2$-sided ranges.
\begin{algorithm}
\caption{Approximate $\ell_1$ Nearest Neighbor using SRT in $\MPC(m)$}
\label{alg:nn}
\begin{algorithmic}[1]
\Require{A point set $P$}
\Ensure{Approximate nearest neighbors in the cones around input points.}
\State{$T=$ A SRT constructed using \Cref{alg:flat} on $P$.}
\State{find the $\ell_1$ nearest neighbors of the corners of the cells of the leaves of $T$.}
\For{each leaf $(i,j)$ of $T$ \textbf{in parallel}}
\State{$Q=Q\;\cup$ compute the range query of nodes $U_{i,j}$ of $T$ in each dimension.}
\EndFor
\State{compute the queries $Q$.}
\State{find the nearest neighbors of all the leaves of $T$ using the queries from $Q$.}
\State{locally compute the nearest neighbors of the points of $P$.}
\end{algorithmic}
\end{algorithm}

\Cref{alg:nn} takes $O(1/\eta)$ rounds and $O(n)$ work since $T$ has height $O(1/\eta)$.

\begin{theorem}\label{theorem:2sided}
All $2$-sided range queries for $\ell_1$ nearest neighbor can be computed in $\MPC(m)$ using \Cref{alg:nn}.
\end{theorem}
\begin{proof}
The tree has $m$ children in each internal node and $m$ points in each leaf, so, the number of nodes in $T$ is $O(n/m)$.
Using \Cref{lemma:semigroup}, nearest neighbor queries for rectangular ranges that appear as the nodes of the tree can be computed by aggregating the values of each subtree in its root, which takes $O(\log^2_m n)$ parallel time in $\MPC(m)$. This gives the nearest neighbor to each corner of the rectangular range defined by each node of the tree.

The rest of the queries can be computed by querying the range tree $T$ similar to the sequential version using the ranges of the $2$-sided query and merging the results by taking the minimum distance to the apex of the query. Each $2$-sided range query can be covered by at most $O(\log^2 n)$ rectangular queries from the tree (set $U_{i,j}$ for the children of node $(i,j)$). There are $O(n\log^2 n)$ queries in $Q$, since each point appears in at most $O(\log^2 n)$ cells.
% solutions in the blocks of nodes of $T$ and $O(m)$ queries in each node, so the total number of answered queries with replication is $O(n \log_m^2 n)$.
%
%Global sorting is required to send the data to the same machines to be computed locally, which takes $O(\log_m (n\log_m^2 n))=O(\log_m n+\log_m \log_m^4 n)$ rounds.
%Since we use blocks of size $m$, the height of the tree is $O(\log^2_m n)$, so the round complexity is $O(\log^2_m n)=O(1/\eta^2)$.
Computing these queries is in $\MPC(m)$ because it ignores polylogarithmic factors in space. Local computations do not require extra rounds.
\end{proof}
\Cref{theorem:2sided} is enough to build a Yao-graph with $4$ cones (using $\ell_1$ distance).  A similar approach to \cite{mine} can generalize this method to use more cones.

\subsubsection{A $\mathbf{MPC}$ Algorithm for Yao-Graph}\label{sec:yao}
\Cref{lemma:inside} proves the stretch factor of Yao-graph for points inside SRT cells.
\begin{lemma}\label{lemma:inside}
Any part of Yao-graph that falls inside a polygon with edges parallel to the sides of the cones of Yao-graph is a spanner for the points inside that polygon with the same stretch factor as the Yao-graph on all the input points.
\end{lemma}
\begin{proof}
Routing in Yao-graph uses the nearest neighbor in each cone, which is either inside the cut region or the other endpoint has been removed. None of these cases change the shortest path between two points.
\end{proof}

In \Cref{alg:wspd}, we give a Yao spanner for the $\ell_1$ distance.
\begin{algorithm}[h]
\caption{Yao-Graph in $\MPC(m)$}
\label{alg:wspd}
\begin{algorithmic}[1]
\Require{A point set $P$, a constant $s$}
\Ensure{The adjacency list $A$ of a spanner graph}
\State{$\epsilon=8/(s-4)$}
\State{$k=\lfloor \frac{\pi}{\sqrt{\epsilon}}\rfloor$}
\State{$P'=\{(\theta,x,y)\mid \forall (x,y)\in P,\; \theta=\lfloor \frac{\tan^{-1}(y/x)}{(2\pi/k)} \rfloor\}$}
\For{$i=0, \ldots, k-1$ \textbf{in parallel}}
\State{Compute the nearest neighbor queries using \Cref{alg:nn} on $P'$.}
\EndFor
\\ \Return {the graph $H=(P,E)$, $E=$ the edge between the nearest neighbors.}
\end{algorithmic}
\end{algorithm}

\begin{theorem}\label{theorem:stretch}
\Cref{alg:wspd} gives a $(1+\epsilon)$-approximation spanner.
\end{theorem}
\begin{proof}
Using \Cref{lemma:inside}, the approximation factor of the shortest paths computed by the Yao-graphs constructed inside the cells depends only on the number of cones, which we set based on $\epsilon$. Using \Cref{lemma:mine},
\[
k= \lfloor \frac{\pi}{\sqrt{\epsilon}} \rfloor \Rightarrow \theta\leq\frac{2\pi}{k}\leq 2\sqrt{\epsilon} \Rightarrow 1+\frac{\theta^2}{8}+o(\theta^2)\leq 1+\epsilon.
\]
\end{proof}

\bibliographystyle{unsrt}
\bibliography{refs}

\end{document}